\documentclass[runningheads]{llncs}
\usepackage[T1]{fontenc}
\usepackage{graphicx}
\usepackage{amsmath,amsfonts,amssymb}
\usepackage{cleveref}
\usepackage{xcolor}

\newcommand\bids[0]{\mathbf{b}}

\makeatletter
\def\@fnsymbol#1{\ensuremath{%
  \ifcase#1\or *\or \dagger\or \ddagger\or \mathsection\or
  \mathparagraph\or \|\or **\or \dagger\dagger\or \ddagger\ddagger
  \else\@ctrerr\fi}}
\makeatother

\begin{document}
\title{Beyond Winner-Take-All Procurement Auctions}
%
%
\author{
Pranav Garimidi\textsuperscript{$*$} \and
Michael Neuder\textsuperscript{$\dagger$} \and
Tim Roughgarden\textsuperscript{$\ddagger$} \\
\vspace{-2em}
}

\institute{}

\authorrunning{Garimidi et al.}

\maketitle              
\footnotetext[1]{a16z crypto. Email: pgarimidi@a16z.com.}
\footnotetext[2]{Princeton University. Author supported in part by the Ethereum Foundation: Grant FY25-2276. Email: michael.neuder@princeton.edu.}
\footnotetext[3]{Columbia University \& a16z crypto. Author’s research at Columbia University supported in part by NSF awards CCF-2006737 and CNS-2212745. Email: tim.roughgarden@gmail.com.}

\begin{abstract}
Blockchain protocols often seek to procure computationally challenging work from a decentralized set of participants. While there are simple procurement auctions that result in the minimal cost of acquisition and maximal efficiency, they also lead to concentration in the provider set due to the winner-take-all market structure. We design and analyze single-good procurement auctions that balance social-cost minimization (at the extreme, a winner-take-all auction) with decentralization (at the extreme, a uniform allocation). We first give a dominant-strategy incentive-compatible (DSIC) mechanism explicitly designed to implement non-winner-take-all allocations. 
Our allocation rule uniquely solves an optimization with respect to a modified social-cost metric that penalizes large, single-player concentrations and is parameterized with a curvature value, $\alpha$, with $\alpha \rightarrow 0$ implementing the uniform allocation and $\alpha \rightarrow \infty$ implementing the winner-take-all allocation. We further quantify the loss in social cost of this mechanism as a function of $\alpha$.


$\qquad$ We then propose two alternative mechanisms, each addressing a limitation of the DSIC mechanism, namely a lack of Sybil-resistance and a complex payment rule. First, we examine a variation of Tullock contests to achieve a non-winner-take-all Sybil-proof procurement mechanism. Second, we consider a mechanism with the same allocation rule as the DSIC mechanism but with an alternative payment rule in which producers are simply paid proportionally to their bids. This provides a much simpler payment rule which, while not DSIC, still results in the mechanism being ex-post ``safe'' (where there exists a bidding strategy that is guaranteed to result in non-negative utility) for participating bidders. For both non-DSIC mechanisms, we characterize the equilibrium allocations and prove price of anarchy bounds.
\end{abstract}%
\section{Introduction}

Many blockchain-related protocols face {\em procurement problems}, in which the protocol must secure services from a permissionless set of providers. For example, layer-one blockchain protocols require validators to secure the network, prover marketplaces gather provers to fulfill requests for succinct proofs, and decentralized inference networks recruit GPU providers to carry out queries. 
Procurement problems are well-studied in the traditional mechanism design literature, where the focus has generally been 
on \emph{efficient} mechanisms that seek to allocate work to parties that can complete the work at the lowest cost.
However, procurement mechanisms that focus solely on minimizing costs lead to \emph{winner-take-all} outcomes, where the single most efficient service provider consistently wins all the work. 

In contrast, blockchain-related applications often additionally require a ``sufficiently decentralized'' provider set. 
For example, a typical PBFT-style consensus protocol requires a sufficiently large number of distinct validators in order to be confident that fewer than one-third of them are faulty. 
Similarly, a prover or inference network requires a sufficient number of active providers to ensure that a single one cannot cause a liveness failure for the protocol's consumers. 

Winner-take-all procurement encourages centralization, with all but the most efficient provider leaving the market or consolidating with the most efficient one. An end state with a single provider would threaten the security assumptions of many blockchain-related applications, and more generally, can be economically undesirable for several reasons:

\begin{itemize}
    \item Once a single service provider has crowded out their competition, it can price-gouge by charging monopoly prices.
    \item Having only a single service provider makes the procurer vulnerable to the provider failing for any number of reasons. 
    \item With no competition, there are no incentives for a single service provider to make investments in decreasing their cost.
    \item In case the type of work demanded changes in the future, it is favorable to have multiple providers with different specialties available. 
    \item Once a single service provider dominates the market, it can undercut costs whenever a new participant attempts to enter the market, defending their monopoly. 
\end{itemize}

In practice, blockchain-related protocols are often designed to maintain a sufficiently decentralized set of providers. For instance, in addition to acting as Sybil-resistance mechanisms, proof-of-work and proof-of-stake convert earning consensus rewards into a Tullock contest, a mechanism in which many nodes participate at equilibrium instead of only the most efficient one~\cite{arnosti2022bitcoin}. Similarly, 
in prover markets, proposals have been made to explicitly trade off efficiency for a more robust set of providers~\cite{succinct}. These examples demonstrate the importance of understanding the design space of procurement mechanisms that go beyond pure cost-minimization. As blockchain protocols and their applications expand to new contexts, these mechanisms will become increasingly important. This work initiates the exploration of this design space 
through the study of three natural mechanisms and their properties, detailed in the next section.

\subsection{Summary of results}



This work is an applied modeling contribution that serves to motivate and discuss trade-offs in the study of procurement in blockchain settings. We model the procurement game and define a set of natural properties that practitioners may find desirable, along with giving mechanisms for achieving different subsets of these properties. We study three mechanisms under slightly different models, which we discuss in more detail below, but summarize here for clarity. We start by considering an incomplete information setting where bidders' costs are private and examine a truthful mechanism that explicitly implements a non-winner-take-all allocation. This mechanism, however, has a couple of clear downsides: the payment rule is both difficult for bidders to interpret and for the protocol to implement, and is not Sybil-proof. Additionally, this mechanism can result in arbitrarily high payments for the protocol, but we give a modification that allows the protocol to bound its costs with the downside of a potentially more centralized allocation. To address these downsides, we consider two additional mechanisms. 
Both mechanisms are non-truthful, and thus we move from the incomplete to the complete information setting, where bidders' costs are common knowledge, to analyze the mechanisms' pure Nash equilibria. The first of these mechanisms is essentially a Tullock contest which is Sybil-proof and allows the protocol to control the total payments it gives out. However, this mechanism may still be difficult for bidders to participate in and may even leave bidders with negative utilities for participating. The second mechanism we consider uses the same allocation rule as the truthful mechanism but with a paid-as-bid payment rule, where bidders are paid exactly proportional to their bids and the amount of work they were allocated. This mechanism is comparatively easier for bidders to reason about, with the downside of not being Sybil-proof and potentially resulting in arbitrarily high payments. A reserve price can be used to cap the total payments, but we only analyze the equilibria without a reserve.


We view the present work as an important first step in discussing the trade-offs and desiderata of decentralized procurement mechanisms, while stopping short of proposing a definitive procurement mechanism for protocols to use. The mechanisms we introduce are analyzed under different informational assumptions and thus not directly comparable. Nevertheless, each of the three mechanisms captures important features of interest to protocol designers and can serve as building blocks for the future. We leave a unified model and a characterization of an optimal mechanism for this setting as important future work. 




\Cref{sec:dsic} introduces $\alpha$-proportional allocation rules ($\alpha$-PARs) which implement allocations minimizing an $\alpha$-scaled social cost metric that penalizes allocations with large single-player concentrations, explicitly optimizing for non-winner-take-all allocations. These allocation rules are monotone and thus truthfully implementable via the corresponding Myersonian payment rules. We characterize the trade-off between social cost (where winner-take-all, implemented by taking $\alpha \rightarrow \infty$, is maximally efficient) and ``decentralization'' (where the uniform allocation rule, implemented by taking $\alpha \rightarrow 0$, is maximally decentralized) for these mechanisms by showing the social cost scales in the number of players, $n$, and the protocol-selected $\alpha$ parameter as $1+(n/\alpha)^{1/\alpha}$.  
This DSIC mechanism has two principal drawbacks that, depending on the details of the application, motivate considering alternative procurement mechanisms. First, the mechanism is not Sybil-proof. Second, its payment rule---the unique payment rule that can be coupled with an $\alpha$-PAR to produce a DSIC mechanism---is complex to implement and hard to interpret for bidders. 

\Cref{sec:tullock} presents a Sybil-proof mechanism that we call a Tullock procurement contest. We shift focus here to the complete-information setting (the standard setting for the analysis of Tullock contests), in which each supplier's cost of production is common knowledge. 
This assumption is reasonable in the context of mechanisms that are run repeatedly with public bids (as in the case of many blockchain-related applications), as market participants may learn each other's costs over time. (See \cite{edelman2007internet} for a more extended discussion of this point in the context of the complete-information analysis of generalized second-price auctions for online advertising.)
Here, we prove the existence and uniqueness of PNE in the game induced by a Tullock procurement contest. We then bound the social cost of the PNE, and show that as the auctioneer increases their budget of rewards, the social cost at equilibrium increases, but the allocation becomes more evenly distributed across providers. 
In particular, as the reward budget grows increasingly large, with $n$ competing bidders, their allocations converge to $1/n$. 

\Cref{sec:PAB} steps back and instead uses a simple to implement and interpret ``paid-as-bid'' payment rule with the same $\alpha$-PAR allocations as the DSIC mechanism. Beyond simplicity, any participant who bids at least their true cost is guaranteed non-negative utility, no matter what the other participants bid (we call this property ex-post safety). Unfortunately, this mechanism is not Sybil-proof. We prove the existence and uniqueness of PNE for this mechanism. Further, we implicitly characterize the equilibrium allocations and the associated worst-case cost vectors. We parameterize this class of mechanisms by $\alpha$, where as $\alpha$ grows increasingly large, the mechanism concentrates an increasing amount of the allocation between the lowest-cost bidders. We show that at equilibrium, every agent's allocation is upper bounded by $1-1/\alpha$. Additionally, we provide numerical results demonstrating how the price of anarchy scales with the number of players and the disparity between their respective costs. 
For example, we show that with 16 participants that are no worse than two-times more costly than the best prover, the price of anarchy of the $\alpha=4$ mechanism is no more than 1.4 (i.e., a 40\% degradation in the social cost). \Cref{sec:conclusion} concludes and outlines future work.

\subsection{Related Work}

\paragraph*{Procurement auctions.}
Procurement or ``reverse'' auctions have a rich literature, finding numerous applications in practice. The breadth of topics is too deep to cover here (see \cite{dimitri2006handbook} for a full treatment), so we focus on settings where the auctioneer sacrifices minimizing their cost to ensure more equitable outcomes.
\cite{dynamic_procurement} considers a setting similar to ours, where they model the cost to the auctioneer of agents leaving the auction if they are not allocated enough work. To address this, they study a threshold allocation rule where the lowest $k$ producer bids are each allocated exactly $1/k$ of the production rights. They prove that setting a lower threshold (thus effectively increasing $k$) can lead to higher revenue in the repeated game despite not being revenue maximizing in the one-shot game (where $k=1$ would be optimal). \cite{che1993design} studies multi-dimensional procurement auctions, where the price and quality of the produced good are considered. This opened a line of literature considering non-price features when deciding how to award government contracts and was grounded in how the Department of Defense handled contracts and negotiations. \cite{kang2022winning} extend this to consider that the US government restricts entry in procurement auctions for many reasons: (i) regulatory restrictions on the producers (e.g., domestic vs. foreign), (ii) small business concerns (e.g., promoting and supporting allocation to businesses that otherwise wouldn't be competitive), and (iii) discretionary reasons, a catchall. The ``small business concerns'' are most relevant to our work and were also studied explicitly in \cite{nakabayashi2013small}, which demonstrated that small business participation in procurement auctions would be reduced by 40\% without a significant portion of the budget being set aside for them. Additionally, they show that the procurement cost would increase significantly without the increased competition facilitated by the small business subsidy.

\paragraph*{Supply-side bidding and supply function equilibria.} 
Our work relates to the operations research literature of firms competing to produce some supply of a homogeneous good, also known as Cournot competition. Most notably, \cite{klemperer1989supply} generalizes Cournot's model to allow producers to specify a more general class of ``supply functions,'' which specify the amount of output the firm will produce as a function of price. They characterize the equilibria of these supply functions under uncertainty of the realized demand shocks. In other words, they answer the question, ``How do firms with identical costs but uncertainty around demand determine how to report their supply functions in equilibria?'' \cite{baldick2004theory} studies supply function equilibria when limiting the producers to linear supply functions, while \cite{karpowicz2007characterization} limits even further to only consider scalar supply functions. Both of these models are similar to ours in that the reported values resemble the per-unit cost of production that the producers strategize over, but our allocation rule is distinct.

\cite{johari2004thesis} and \cite{johari2011parameterized} are the most relevant to our work. These works examine supply-side bidding in the ``inelastic demand'' setting, where the demand is relatively fixed, and the producers have asymmetric costs. \cite{johari2004thesis} points out that the simple mechanism that pays each bidder exactly their bid and allocates the resource proportionally (exactly the reverse of the Tullock contest) cannot have a Nash equilibrium, as the bids can grow to infinity while each firm's respective cost remains bounded, and uses this to motivate a constrained set of supply functions. They then solve for the optimal allocation if the firms are price takers, along with the Nash equilibrium when the firms are price anticipating, thus calculating the efficiency loss of the strategic setting. Our work uses the same problem statement but differs in our allocation and payment rules.

\section{Model} \label{sec:model}
We consider a standard procurement setting where there is a divisible unit of work to be assigned to suppliers. There are $n$ competing agents, indexed by $i$, where each agent incurs a cost linear in their allocation parameterized by $c_i$. This cost can be interpreted as the resource cost agent $i$ has to pay to complete a given amount of work. 
Let $\mathbf{c} = (c_1, c_2, \ldots, c_n)$ denote the vector of agent costs. 
We consider both the \emph{incomplete-information} case, where costs are only privately known with competing bidders having no prior knowledge of the costs of their competitors, and the \emph{complete-information} case, where $\mathbf{c}$ is common knowledge to all bidders. 


A procurement mechanism $\mathcal{M}$ consists of an allocation and payment rule ($x,p$). The allocation rule determines how the work is split among the bidders, and the payment rule dictates how much each bidder gets paid. Each agent $i$ simultaneously submits a bid $b_i$, or potentially multiple bids, to the mechanism. We denote the bid vector $\bids = (b_1,b_2,...,b_n)$ and use $\mathbf{b}_{-i}$ to denote the bid vector excluding bidder $i$'s bid. Agents have standard quasi-linear utilities. 

\begin{definition}[Player utility function]\label{def:playerutility}
The player utilities are their payments less the cost incurred from their allocation, $u_i(\bids) = p_i(\mathbf{b}) - c_i x_i(\mathbf{b})$.
\end{definition}

Based on the information agents have about each other, we consider different equilibrium concepts. In the incomplete-information case, where bidders don't have information about each other's costs, we insist that the mechanism be dominant strategy incentive compatible (DSIC), as otherwise, agents lack information on how they should bid. 

\begin{definition}[Dominant-strategy incentive-compatible] A mechanism is DSIC if for every bidder $i$, for all $b_i\ge 0$, and all bid vectors $\mathbf{b}_{-i}$, $u_i(c_i,\mathbf{b}_{-i})\ge u_i(b_i,\mathbf{b}_{-i})$
\end{definition}

In the complete-information setting, where the full cost vector $\mathbf{c}$ is common knowledge to all the bidders, we relax our solution concept to pure-strategy Nash equilibria (PNE).  


\begin{definition}[PNE] A bid vector $\widetilde{\bids}$ is a PNE, if for every bidder $i$, and all alternate bids $b'_i\ge 0$, $u_i(\widetilde{\bids}) \ge u_i(b'_i,\widetilde{\bids}_{-i})$
\end{definition}

We further include two ideal properties for mechanisms to satisfy. First, we target Sybil-proof mechanisms, where agents cannot increase their utility by submitting bids under multiple identities. While this problem can be addressed by the auctioneer knowing the bidder's identities in traditional settings, this property is especially important for permissionless mechanisms where the protocol might not have any information on who the bidders are. Formally, a mechanism is Sybil-proof if every bidding strategy that submits multiple bids is weakly dominated by a bidding strategy that only submits a single bid. 

\begin{definition}[Sybil-proof]
    A mechanism is Sybil-proof, if for all bidding strategies $\sigma_i(\mathbf{c})$, potentially submitting multiple bids, there exists an alternative bidding strategy $\sigma'_i(\mathbf{c}):\mathbb{R}_+^{n}\rightarrow\mathbb{R}_+$ such that for all other bids $\mathbf{b}_{-i}$,  $u_i(\sigma'_i(\mathbf{c}),\mathbf{b}_{-i}) \ge u_i(\sigma_i(\mathbf{c}),\mathbf{b}_{-i})$
\end{definition}


Second, we desire mechanisms that are ex-post ``safe,'' where bidders have a non-trivial bidding strategy (i.e., the strategy doesn't guarantee an agent zero utility regardless of others' bids) under which they never regret participating in the auction. This is similar in spirit to individual rationality, where bidders are certain to have non-negative utility from reporting truthfully, but extended to settings where the bidders are playing strategically. For example, a standard first-price procurement auction is ex-post safe, since as long as an agent makes a bid above their cost, they are guaranteed to have a non-negative utility. Mechanisms that are not ex-post safe have higher barriers to entry, as bidders have to carefully reason about other bidders' strategies to avoid losing money and thus may not enter at all. 

\begin{definition}[Ex-post safe]
    A mechanism is ex-post safe, if for every agent $i$, there exists a bidding strategy $\sigma({c_i})$ such that for all bid vectors $\mathbf{b}_{-i}$, $u_i(\sigma(c_i),\mathbf{b}_{-i})\ge 0$ and there exists a bid vector $\mathbf{b}_{-i}$ such that $u_i(\sigma(c_i),\mathbf{b}_{-i}) > 0$.
\end{definition}


We judge the efficiency of a mechanism by the social cost it incurs, i.e., the total cost incurred by all agents who were allocated work. This is the objective procurement mechanism that typically seeks to minimize. While our mechanisms have goals beyond minimizing social cost, we still use it as a benchmark to quantify \emph{how} much efficiency our mechanisms give up for the sake of a more robust supplier set.  

\begin{definition}[Social cost]
    The social cost of an allocation is the total cost incurred by the agents in an allocation, $SC(\mathbf{c}, \mathbf{x}) = \mathbf{c} \cdot \mathbf{x} = \sum_{i=1}^n c_ix_i$. 
\end{definition}

To quantify how efficient a given mechanism is, we benchmark the social cost under equilibrium bidding versus the optimal social cost achievable. This ratio is referred to as the price of anarchy. Note that in our model, the optimal social cost will always be to allocate all the work to the agent with the lowest cost. 

\begin{definition}[Price of anarchy (PoA)] Let $\mathcal{X}$ be the set of all allocations and $\mathcal{X}_{eq}$ be the set of allocations induced at all pure Nash equilibria (PNE). The price of anarchy is the ratio of the highest-cost equilibrium allocation to the lowest-cost allocation,
\begin{align*}
    PoA = \frac{\max_{\mathbf{x} \in \mathcal{X}_{eq}} SC(\mathbf{c}, \mathbf{x})}{\min_{\mathbf{x} \in \mathcal{X}} SC(\mathbf{c}, \mathbf{x})}.
\end{align*}
\end{definition}

We evaluate the three procurement mechanisms in  \Cref{sec:dsic}, \Cref{sec:tullock}, \& \Cref{sec:PAB}, respectively, along these axes. 

\section{DSIC Mechanism}\label{sec:dsic}
We first consider the incomplete-information setting, where agent costs are private. We start by defining a class of allocation rules parameterized by $\alpha$,  closely resembling those of a Tullock contest \cite{tullock2008efficient} with negative exponents, which seek to explicitly allocate work to higher-cost agents. 

\begin{definition}[$\alpha$-PARs]\label{def:allocation_rule}
    An $\alpha$-\textit{proportional allocation rule} takes the form,
    \begin{align*}
        x_i(\mathbf{b}) = \frac{b_i^{-\alpha}}{\sum_{j=1}^n b_j^{-\alpha}},
    \end{align*} 
    for $\alpha > 0$.
\end{definition}


Recall in the procurement setting that bids correspond to costs, so bidding \textit{higher} results in a \textit{lower} allocation, thus the need for a \textit{monotone decreasing} allocation rule in $b_i$. Given that the allocation rule is monotone, we can implement it with a DSIC mechanism using Myerson's lemma \cite{myerson1981optimal} using the standard normalization that $\lim_{b_i \to \infty}p_i(b_i,\mathbf{b}_{-i}) = 0$.
Letting $D_i=\sum_{j\neq i} b_i^{-\alpha}$, the integral form of the payment rule is,
\begin{align*}
    p_i(b_i, \bids_{-i}) = \frac{b_i}{1+D_ib_i^{\alpha}} + \int_{b_{i}}^{\infty} \frac{1}{1+D_it^{\alpha}}dt. 
\end{align*}
This payment rule represents the amount that the protocol needs to \textit{pay} the agents to elicit truthful reports $\bids$. Note that under this mechanism, the protocol can be forced to pay an arbitrarily large amount. To cap the total payment, we can modify the mechanism to only allocate work to the $k$ lowest bids with a maximum bid of $\bar{b}$ and the payment rule modified accordingly. This also improves the social cost of the mechanism as allocation gets moved only towards the agents with the lowest costs. Thus, social cost and total payments can be traded off with decentralization according to the mechanism designer's preference. The following analysis can then be interpreted as assuming $c_i \le \bar{b}$ for all $i$ and taking $n=k$.  


We quantify exactly how these allocation rules trade off efficiency for distribution by showing that $\alpha$-PARs are the unique solutions to minimizing the objective $\sum_{i=1}^n c_ix_i^{1+1/\alpha}$ via a standard Lagrangian argument (See \Cref{app:dsic-proof} for details).

\begin{lemma}[Optimization problem]\label{lemma:opt}
	The allocation rule $x_i = c_i^{-\alpha} / \sum_{j=1}^n c_j^{-\alpha}$ is the solution to the following constrained optimization problem,
	\begin{align*}
		\min_{\mathbf{x} \in [0, 1]^n}&\sum_{i=1}^n c_ix_i^{1+1/\alpha}\\
		\text{s.t.,} &\sum_{i=1}^n x_i = 1.
	\end{align*}
\end{lemma}

Since $x_i \le 1$ and $\sum_{i=1}^n x_i=1$, the additional $1/\alpha$ term in the exponent pushes the $x_i$ closer together rather than placing all the weight on the lowest-cost agent. Another interpretation is that the function penalizes concentrated allocations by scaling their weight super-linearly in this objective. Thus, we can see how $\alpha$ can be tuned to trade off efficiency and a decentralized provider set since as $\alpha \rightarrow \infty$, we get that the objective converges to the standard social cost benchmark that is minimized by a winner-take-all allocation. 

We now bound the efficiency we trade off by optimizing for a more decentralized set of providers compared to allocating all the work to the lowest cost agent. 
\begin{figure}
    \centering
    \includegraphics[width=\linewidth]{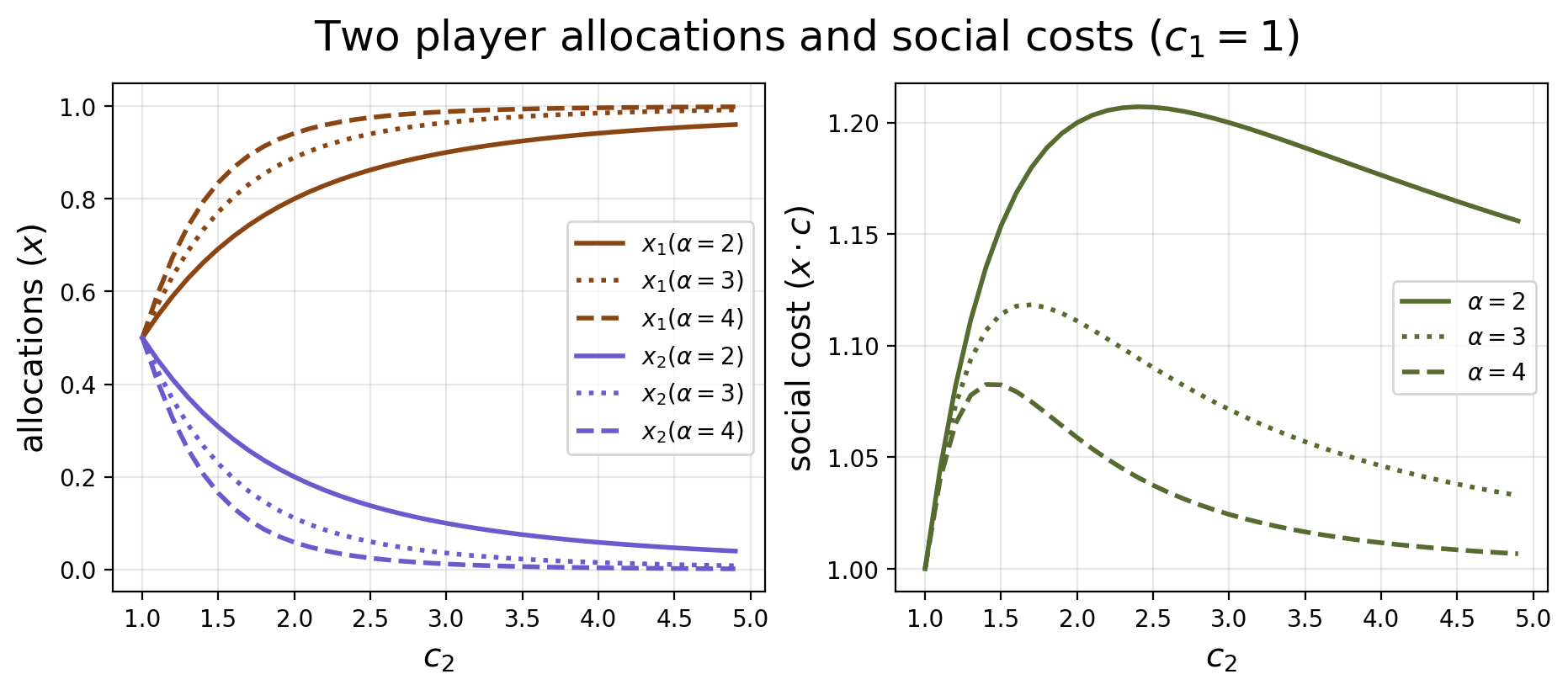}
    \caption{Allocations (left) and social cost (right) with $n=2$ and $c_1=1$ under the DSIC mechanism for various values of $\alpha$.}
    \label{fig:two-player-dsic}
\end{figure}\Cref{fig:two-player-dsic} shows the allocations ($x_1, x_2$) for the DSIC mechanism in the two-player game. We fix $c_1=1$ and vary $c_2$. We also plot the social cost, which for the two-player game is simply $c_1x_1 + c_2x_2.$ Notice that the minimal social cost in this setting would arise from fully allocating the work to player one, leading to a total incurred cost of $1$. Also, notice how the value of $\alpha$ shapes the rate of concentration of the allocation to the low-cost player. Since the mechanism is truthful, the low-cost player will continue bidding $c_1=1$, while the high-cost player will bid $c_2$. As $c_2$ increases, $c_1^{-\alpha}$ dominates the sum, and the full allocation goes to player one. Lastly, the social cost is not monotone increasing in $c_2$. 

For the $n$-player game, the social cost increases in $n$ because a larger amount of the allocation is given to higher-cost players. The rate at which the social cost degrades is highly dependent on $\alpha$. For large values of $\alpha$, the allocation heavily concentrates around the lowest-cost agent regardless of how many other agents there are. \Cref{thm:dsic-worst-case} bounds this scaling precisely. We show that the instance with the worst-case social cost takes the form of $[1,r,r,\ldots, r]$ for some $r>1$ dependent on $n$ and $\alpha$. Intuitively, this setting implies that there is one ``efficient'' player and $n-1$ ``inefficient'' players. From this characterization, we find the worst-case $r$ to give a tight bound on the worst-case social cost. 
\begin{theorem}\label{thm:dsic-worst-case}
     The DSIC mechanism (the $\alpha$-PAR allocation rule and the Myersonian payment rule) with $n$ players has a worst-case social cost approximation of $1+\left(\frac{n}{\alpha}\right)^{\frac{1}{\alpha}}$ for $\alpha \ge 1$.
\end{theorem}
\begin{proof}
    Note that both the allocation rule and social cost-objective are homogeneous functions of degree zero in $\mathbf{c}$. Thus WLOG assume for all agents $i$, $c_i \le c_{i+1}$ with $c_1=1$. In this case, the optimal social cost is allocating the entire job to agent 1 for a cost of 1. Thus, to show the approximation, it suffices to bound the social cost of any instance with $c_1=1$. Since the mechanism is DSIC, we can take each $x_i = \frac{c_i^{-\alpha}}{\sum_{j=1}^n c_j^{-\alpha}}.$ 
    This implies the social cost under a cost vector of $\mathbf{c}$ is 
    \begin{align*}
        SC(\mathbf{c}, \mathbf{x}) = \sum_{i=1}^n c_ix_i = \frac{\sum_{i=1}^n c_i^{1-\alpha}}{\sum_{i=1}^n c_i^{-\alpha}}.    
    \end{align*}
    Notationally, we now refer to the social cost only in terms of the cost vector, $SC(\mathbf{c})$. Given this, we can show that for $\alpha>1$, the cost vector maximizing the social cost must be of the form $[1,r,...,r]$. The argument follows from showing that the first order conditions $SC(\mathbf{c})$ are only satisfied when $c_i=c_j$ for all $i,j>1$. Using this, we let  $f(r)=\frac{1+(n-1)r^{1-\alpha}}{1+(n-1)r^{-\alpha}}$ be the social cost given a cost vector $[1, r, r, \ldots, r]$ and the theorem follows mechanically via maximizing $f(r)$ over all $r\ge 1$ (See \Cref{app:dsic-proof} for details).\hfill \qed
\end{proof}

This bound shows that the social cost is not strictly increasing in each agent's costs. There is a certain threshold beyond which an agent's increasing cost decreases their allocation at a faster rate than their increased cost harms efficiency. 
 

While this parameterization lets a mechanism designer precisely trade off between decentralization and efficiency, it has two key issues: (i) it is not Sybil-proof, and (ii) it is hard for bidders to reason about and difficult to implement on-chain. To see that the DSIC mechanism is not Sybil-proof, notice that a bidder can simply make infinite bids with a constant cost of $C$. Each additional bid increases their allocation at the same cost.

Regarding the complexity of the payment rule, while the mechanism is DSIC, there is a large empirical literature showing that bidders don't play dominant strategies in DSIC mechanisms even when the mechanism is relatively simple, such as a second-price auction \cite{kagel1993independent,harstad2000dominant,garratt2012behavior}. Furthermore, the payment rules required to implement the allocation rules as described would be quite complex and require numerical approximations, leading to complex code for the auction mechanism. 
One would expect bidders to fail to understand that these mechanisms are truthful and, hence, bid sub-optimally. In \Cref{sec:tullock}, we first give a mechanism that is Sybil-proof but still hard for bidders to reason about; in \Cref{sec:PAB}, we give a mechanism that is also not Sybil-proof but easier for bidders to participate in. 

\section{Tullock Procurement Contests} \label{sec:tullock}
To construct a Sybil-proof mechanism that avoids winner-take-all equilibria, we take inspiration from the literature on Tullock contests. Tullock contests refer to a class of games where agents compete for a good via making costly investments \cite{tullock2008efficient}. The good is typically allocated randomly proportional to agents' investments, but all agents have to pay regardless of whether they win or not. Usually, in Tullock contests, the loss in efficiency from a suboptimal agent winning the good is a downside. However, in our case, we will use this fact to implement a non-winner-take-all equilibrium. We turn our setting into a Tullock contest by having the procurer put up a budget $B$ as a reward for bidders to compete over, which is presented as a ``Proof Contest'' in \cite{succinct}. Agents then submit bids, and both the work and budget are allocated proportionally to agents' bids. However, agents get their bid deducted from their reward regardless of how much reward they are allocated. Under this mechanism, we allow bidders to have negative utilities where they potentially have to pay the auctioneer. 

\begin{definition}[Tullock procurement contest (see also \cite{succinct})]
    A Tullock procurement contest is given by the following allocation  and payment rule parameterized by a budget $B$:
    \begin{align*}
        x_i(\bids) = \frac{b_i}{\sum_{j=1}^n b_j}, \ \ p_i(\bids) = Bx_i(\bids)-b_i.
    \end{align*}
\end{definition}
This implies that agents have utilities given by 
\begin{align*}
    u_i(\bids) = \frac{b_i}{\sum_{j=1}^n b_j} (B - c_i) - b_i 
\end{align*}
From this, we can see that Tullock contests are Sybil-proof, as a bidder can always submit a single bid to get the same payoff they would have from multiple bids.  

\begin{lemma}
    Tullock procurement contests are Sybil-proof.
\end{lemma}
\begin{proof}
    Consider a bidding strategy $\sigma_i(c) = (b_i^1,...,b_i^k)$. Then let $\sigma'_i(c) = \sum_{j=1}^k b_i^j$. Let the sum of the other agents' bids be $\beta$. Then we have that 
    \begin{align*}
        u_i(\sigma_i(c),\mathbf{b}_{-i}) = \sum_{j=1}^k \left(\frac{ b_i^j}{\sum_{\ell=1}^k b_i^\ell+\beta}(B-c_i)- b_i^j\right) = u_i(\sigma'_i(c),b_{-i}).
    \end{align*} 
    \hfill \qed
\end{proof}

We now characterize the equilibrium of these mechanisms and how they achieve a near-optimal price of anarchy (PoA). We start by recalling the standard Tullock contest setting and well-known results characterizing their equilibrium and PoA. 

\begin{definition}[Tullock contests (from \cite{tullock2008efficient})]
    A Tullock contest is defined by a set of $n$ agents, each having value $v_i$ for some good. Each agent simultaneously submits a bid $b_i$, and the good is randomly allocated proportional to bids, with every agent paying their bid regardless of whether they won. 
\end{definition}

Notice that a Tullock procurement contest with a budget of $B$ is a standard Tullock contest where each agent $i$ has value $\max\{B-c_i,0\}$ for winning the item. Thus, we can use the following well-known results characterizing the PNE of Tullock contests to help analyze the equilibrium for Tullock procurement contests.

\begin{figure}
    \centering
    \includegraphics[width=0.9\linewidth]{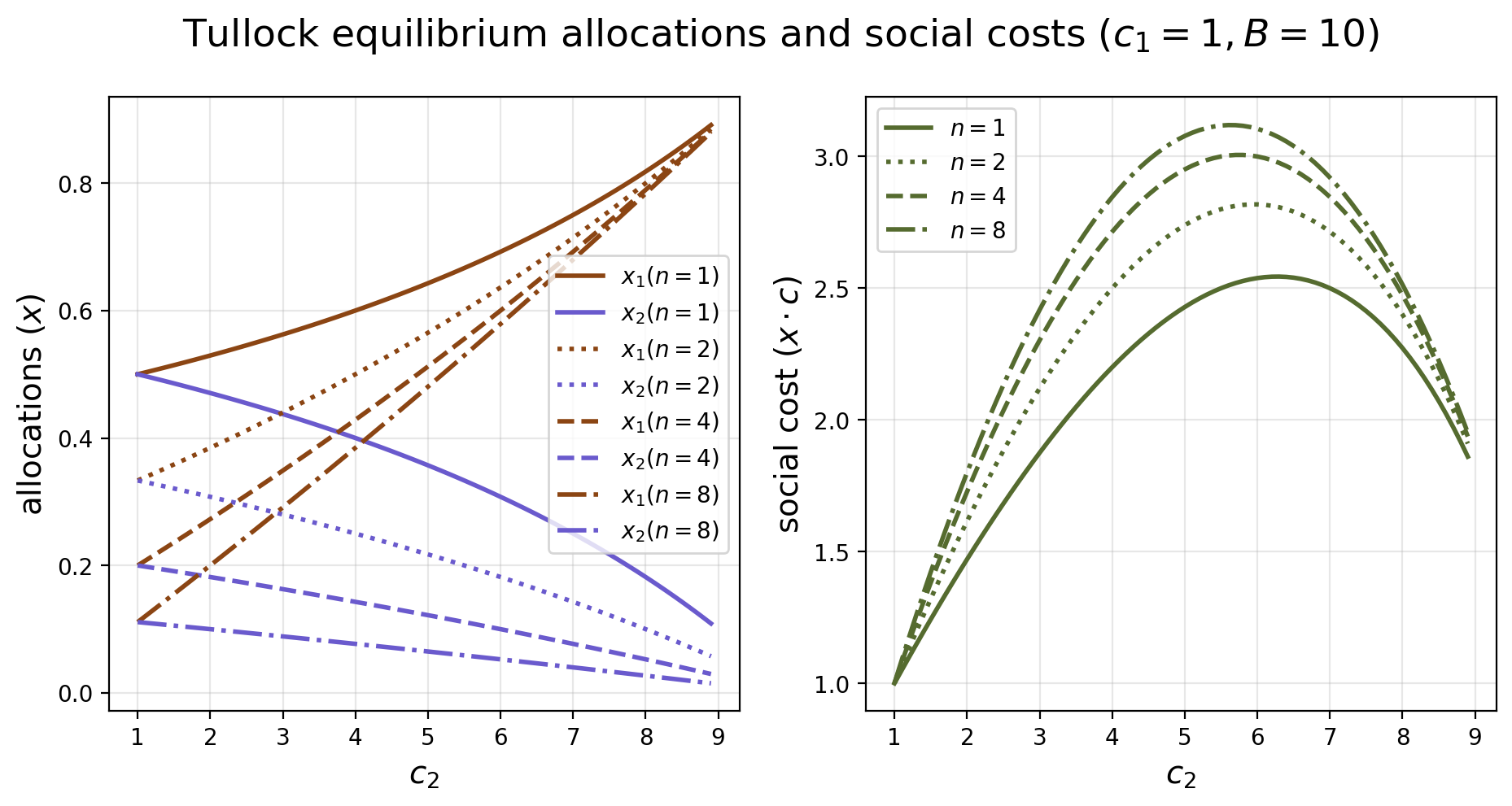}
    \caption{Equilibrium allocations and the resulting social costs in the Tullock procurement mechanism as a function of $c_2$ with $c_1=1$ and various values of $n$.
    Note that we only show $x_2$ in the left subplot because $x_2=x_3=\ldots=x_n$.}
    \label{fig:n-player-tullock}
\end{figure}

\begin{lemma}\label{l:tullock_pne_exist}(from \cite{Tullock-Ch7})
     Tullock contests where at least two agents have strictly positive value for the good have unique PNEs.
\end{lemma}

\begin{lemma}\label{l:tullock_PoA}(from \cite{johari2004efficiency})
    At the equilibrium allocation $\tilde{x}$ induced by the PNE of a Tullock contest, $\sum_{i=1}^n v_i\tilde{x}_i \ge \frac{3}{4} \max\{v_i\}$. 
\end{lemma}

\begin{lemma}\label{l:tullock_pne_char}(from \cite{arnosti2022bitcoin})
    Let $F(x) = \sum_{i}\max\left(1-\frac{x}{v_i},0\right)$ and let $v^*$ be the solution to $F(v^*)=1$. Then the equilibrium allocation $\tilde{x}$ of a Tullock contest is given by $\tilde{x
    }_i = \max\left(1-\frac{v^*}{v_i},0\right)$
\end{lemma}

\Cref{fig:n-player-tullock} shows the equilibrium allocations and resulting social costs for the Tullock contest as a function of $c_2$ for the cost vector $\mathbf{c} = [1, c_2, \ldots, c_2].$ Various values of $n$ demonstrate the different rates at which the allocation concentrates on player 1. The right subplot shows that as $n$ increases, the maximal social cost increases. Intuitively, this arises because larger portions of the allocation are given to the high-cost players.

From \Cref{l:tullock_pne_exist}, we have that as long as there are two agents such that $B>c_i$, the Tullock procurement contest has a PNE. Then, applying \Cref{l:tullock_PoA}, we can bound the PoA of Tullock procurement contests as a function of $B$ and $c_{min}$ where $c_{min}$ is the lowest cost of any agent. 

\begin{lemma}
    Let $\gamma = B/c_{min}$. Then the PoA of the Tullock procurement contest is $\frac{\gamma+3}{4}$
\end{lemma}
\begin{proof}
    Let $\tilde{x}$ be an allocation induced by the equilibrium of the Tullock contest. From \Cref{l:tullock_PoA}, we have 
    \begin{align*}
        \sum_{i=1}^n (B-c_i)\tilde{x}_i \ge \frac{3}{4} (B-c_{min}) \implies \sum_{i=1}^n c_i\tilde{x}_i \le \frac{B+3c_{min}}{4}
    \end{align*}
    The lemma then follows from the optimal social cost being $c_{min}$.
    \hfill \qed
\end{proof}

We see that, holding the costs constant, the PoA is minimized when $B$ is slightly above the second lowest cost out of all the agents. This corresponds to only the two bidders with the lowest costs bidding when the total reward is sufficiently small. On the contrary, as $B$ grows large, even the highest-cost bidders participate. For sufficiently large $B$, all bidders receive an approximately equal allocation at equilibrium as the ratio of their values shrinks. Thus, a protocol can increase $B$ to trade off between efficiency, total payment, and how many bidders they desire to be competitive at equilibrium.

\begin{lemma}
    In a Tullock procurement contest with equilibrium allocation $\tilde{x}$, for all agents $i$,  $\lim_{B\rightarrow\infty} \tilde{x}_i = \frac{1}{n}$
\end{lemma}

\begin{proof}
    We first show that for large enough $B$, all agents bid and hence have a positive allocation at equilibrium. Consider the function $F(B,x) = \sum_{i=1}^n \max\{1-\frac{x}{B-c_i},0\}$ modified from \Cref{l:tullock_pne_char} for our setting. Define $x^*(B)$ to be the solution to $F(B,x^*(B))=1$. It suffices to show $B-c_{max} > x^*(B)$ for sufficiently large $B$. Note that for all $i$, $\lim_{B\rightarrow \infty} \frac{B-c_{max}}{B-c_i} =1$, thus $\lim_{B\rightarrow \infty} \sum_{i=1}^n 1-\frac{B-c_{max}}{B-c_i} < 1 $ and $B-c_{max} > x^*(B)$ since $F(B,x)$ is monotone decreasing in $x$. From this, for sufficiently large $B$, we have that all bidders are active, allowing us to explicitly solve for $x^*(B)$: 
    \begin{align*}
        \sum_{i=1}^n 1 - \frac{x^*(B)}{B-c_i}=1 \implies x^*(B)= \frac{n-1}{\sum_{i=1}^n \frac{1}{B-c_i}}.
    \end{align*}
    It follows that $\lim_{B\rightarrow \infty} \frac{x^*(B)}{B-c_i} = \frac{n-1}{n}$. Then again using \Cref{l:tullock_pne_char} again, we have that 
    \begin{align*}
        \lim_{B\rightarrow \infty} \tilde{x}_i = \lim_{B\rightarrow \infty} 1 - \frac{x^*(B)}{B-c_i} = \frac{1}{n}. 
    \end{align*}
    \hfill \qed
\end{proof}

Note that while Tullock contests are Sybil-proof and implement near-optimal equilibria alongside decentralized allocations, bidders may find them problematic. In particular, Tullock contests are not ex-post safe. Regardless of a bidder's cost and the budget, for any non-zero bid $b_i$, there exists a bid vector $\bids_{-i}$ such that $(B-c_i)x_i(\bids) = (B-c_i)\frac{b_i}{b_i +\sum_{j=1}^n b_j} < b_i$ hence giving them a negative utility. This means that bidders participating in such a mechanism must be sophisticated enough to reason about others' costs and bids to avoid losing money by participating. 

In the following section, we will give a mechanism that is ex-post safe and has a simple payment rule, alleviating the difficulty of participation. However, this mechanism is no longer Sybil-proof, making it unsuitable for many decentralized applications. 

\section{Paid-as-Bid Mechanism}\label{sec:PAB}
In this section, we consider a hybrid mechanism between the DSIC and Tullock contest mechanisms of the previous two sections. In particular, we use the same $\alpha$-PAR
allocation rule (\Cref{def:allocation_rule}) but a simpler payment rule inspired by the all-pay-as-bid Tullock rule. Since $\alpha$-PARs are reverse auction allocation rules, we consider the payment rule where the bidder is \textit{paid} proportionally to their bid. These mechanisms are motivated by simplicity of implementation, where it is clear to users how their payment was derived from their bid and allocation. This is in contrast to the DSIC mechanism, where agents find it hard to reason about their payments even though truthful bidding is a dominant strategy.

\begin{definition}[Paid-as-bid payment rule]\label{def:payment-proportional}
	The \textit{paid-as-bid} rule pays each producer $i$, $p_i = b_i x_i.$
\end{definition}

Plugging the payment rule into the player utility function (\Cref{def:playerutility}), we can rewrite $i$'s utility as,
\begin{align*}
    u_i(b_i, \mathbf{b}_{-i}) := (b_i - c_i) x_i(\mathbf{b}).
\end{align*}

As the mechanism is paid-as-bid (rather than the Myersonian payment rule), it will not be DSIC. Thus, we consider the PNE as the solutions of interest for various cost vectors, c, in the complete-information setting. We start by showing that PNE always exists for this mechanism regardless of the private costs.

\begin{theorem}[Unique PNE]\label{thm:unique_pne}
    The mechanism defined by an $\alpha-$PAR with $\alpha > n/(n-1)$ and the proportional-to-bid payment rule has a unique PNE characterized by the bid vector $\mathbf{\tilde{b}}$ that solves the implicit equation
    \begin{align*}
        \tilde{b}_i = \frac{\alpha c_i(1-x_i)}{\alpha(1-x_i)-1}.
    \end{align*}
\end{theorem}
\begin{proof}
Let $\mathbf{b}=(b_1,...,b_n)$ be a bid vector and let $o_i = \sum_{j\neq i} \frac{1}{b_j^\alpha}$. Then we can rewrite $u_i(b_i,\mathbf{b}_{-i}) = (b_i-c_i)\left(\frac{1}{1+b_i^\alpha o_i}\right)$. From this, we can calculate the first order conditions for $\mathbf{b}$ to be an equilibrium bid vector: 

\begin{align*}
    \frac{\partial u_i(b_i,\mathbf{b}_{-i})}{\partial b_i} = \frac{1+b_i^\alpha o_i+\alpha b_i^{\alpha-1}o_i(c_i-b_i)}{(1+b_i^\alpha o_i)^2} = 0 \\
    \iff 1+b_i^\alpha o_i+\alpha b_i^{\alpha-1}o_i(c_i-b_i)=0.
\end{align*}

Note that this expression is strictly increasing for $0\le b_i < c$ and strictly decreasing for $b_i>c$. Since $u'(0,\mathbf{b}_{-i})=1$ this implies $\frac{\partial u_i}{\partial b_i}$ has a unique zero corresponding to a global max for $u_i(b_i,\mathbf{b}_{-i})$. This implies satisfying the first-order conditions is sufficient for $b$ to be an equilibrium. Then, from the definition of $o_i$ and our allocation rule $x$, we have 
\begin{align*}
    \frac{1}{1+b_i^\alpha o_i} = x_i \iff  b_i^\alpha o_i =\frac{1}{x_i}-1.    
\end{align*}

Substituting this into our first-order condition gives 
\begin{equation}\label{eq:equilibrium_bids}
    \frac{1}{x_i}+\frac{\alpha c_i}{b_i}\left(\frac{1}{x_i}-1\right)-\frac{\alpha}{x_i}+\alpha = 0 \iff b_i = \frac{\alpha c_i(1-x_i)}{\alpha(1-x_i)-1}
\end{equation}

Letting $B=\sum_{i=1}^n b_i^{-\alpha}$, and noting $x_i = \frac{b_i^{-\alpha}}{B} \implies b_i=(x_iB)^{-1/\alpha}$ this gives 
\begin{equation*}
    \frac{\alpha c_ix_i^{\frac{1}{\alpha}}(1-x_i)}{\alpha(1-x_i)-1} = \frac{1}{B^\frac{1}{\alpha}} \ \ \forall i\in[n].
\end{equation*}

It follows that for $x$ to denote an equilibrium allocation
\begin{equation}\label{eq:equil_condition}
    \frac{ c_ix_i^{\frac{1}{\alpha}}(1-x_i)}{\alpha(1-x_i)-1} = \frac{ c_jx_j^{\frac{1}{\alpha}}(1-x_j)}{\alpha(1-x_j)-1} \ \ \forall i,j\in[n].
\end{equation}

Assume that such an allocation $\tilde{x}$ satisfying (\ref{eq:equil_condition}) exists, then we show that setting $\mathbf{\tilde{b}}$ according to (\ref{eq:equilibrium_bids}) with $x_i=\tilde{x}_i$ is consistent with Definition \ref{def:allocation_rule}. Let $\frac{\alpha c_i\tilde{x}_i^{\frac{1}{\alpha}}(1-\tilde{x_i})}{\alpha(1-\tilde{x}_i)-1} = \kappa$. Then from manipulating (\ref{eq:equilibrium_bids}) we get $\tilde{x_i} = (\alpha\kappa/\tilde{b}_i)^\alpha$. Combining this with  $\sum_{i=1}^n \tilde{x}_i=1$ gives $(\alpha\kappa)^\alpha = 1/\sum_{i=1}^n \tilde{b}_i^{-\alpha}$ in turn showing $\tilde{x}_i = \tilde{b}_i^{-\alpha}/\sum_{j=1}^n \tilde{b}_j^{-\alpha}$ as desired. 

Now, to show that a unique PNE exists, it is sufficient to show that there is a unique allocation $\tilde{x}$ satisfying (\ref{eq:equil_condition}). Define the function 
\begin{equation}\label{eq:f_def}
f(x):= \frac{x^\frac{1}{\alpha}(1-x)}{\alpha(1-x)-1}.
\end{equation} Then an allocation vector satisfying equation \ref{eq:equil_condition} is equivalent to finding a vector in the unit simplex $x\in \Delta^{n-1}$ such that $c_if(x_i)=c_jf(x_j) \ \forall i,j\in [n]$. Further note that for $x\in[0,1-1/\alpha)$, $f(x)$ is strictly increasing with $f(0)=0$ and $\lim_{x\rightarrow 1-1/\alpha} f(x) =\infty$. It follows that the inverse of $f$, $f^{-1}:[0,\infty)\rightarrow [0,1-1/\alpha)$ is well defined. Thus for a fixed $d\in[0,1-1/\alpha)$, let $g_i(d) = f^{-1}\left(\frac{c_1}{c_i}f(d)\right)$ with $g(d)=\sum_{i=1}^n g_i(d)$. $f$ being continuous and strictly increasing implies that $g$ is continuous and strictly increasing. Thus, $g(0)=0$ and $g(1/n) \ge 1$  implies there exists a unique $d^* \in (0,1/n]$ such that $\sum_{i=1}^n g_i(d^*)=1$. Note that $g(1/n)$ is well defined since $\alpha>n/(n-1)$. We can then take $\tilde{x}_i = g_i(d^*) \ \forall i\in[n]$ to get the unique equilibrium allocation. \hfill \qed
\end{proof}

\Cref{fig:two-player-bids-allos-equil} shows the equilibrium bids and allocations, denoted $\tilde{b}$ and $\tilde{x}$, respectively, for the two-player game. With $c_1=1$, we increase $c_2$ to demonstrate how the equilibrium bids and allocations evolve. For comparison, we include the honest bids and allocations (shown as solid lines in each subplot) arising from the truthful bidding in the DSIC mechanism from \Cref{sec:dsic}.

\begin{figure}
    \centering
    \includegraphics[width=\linewidth]{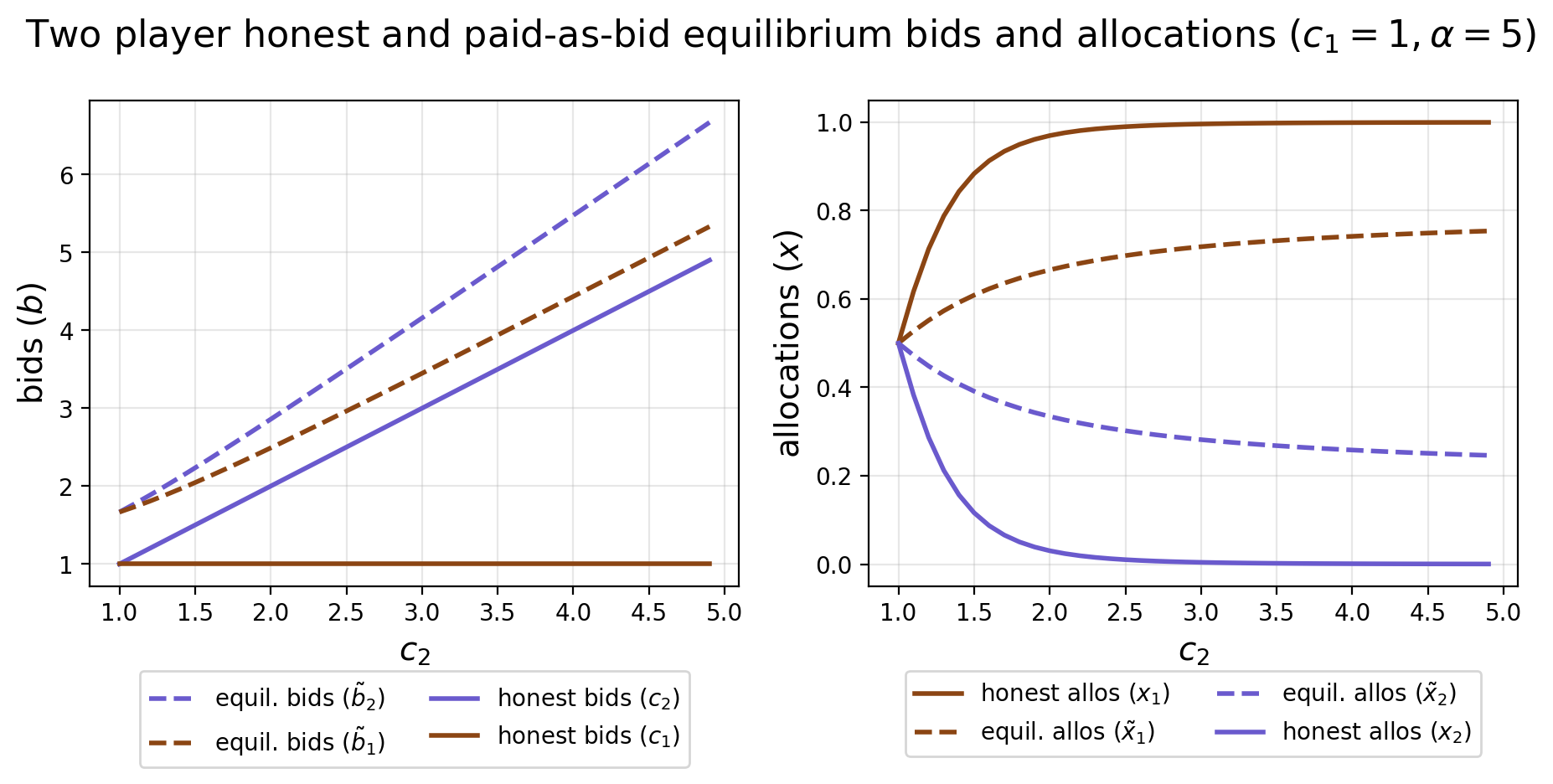}
    \caption{Honest and equilibrium-induced bids (for the paid-as-bid payment rule, \Cref{def:payment-proportional}) and allocations for the two-player game with $\alpha=5$. We see that the equilibrium allocations approach their limits of $1/5$ and $1-1/5$ as $c_2$ increases. Also note that the equilibrium bids for player one increase in $c_2$, while honest bids remain constant at $c_1=1$.}
    \label{fig:two-player-bids-allos-equil}
\end{figure}

The exact equilibrium bids and allocation can be computed numerically, given this implicit characterization for different cost vectors. A critical feature of the allocations at equilibrium is that they are bounded based on the size of $\alpha$ as $x_1 < 1-1/\alpha$ and $x_2 > 1/\alpha$, which we describe in the following corollary.

\begin{corollary}[Bounded allocations]\label{corr:bounded}
	Under the equilibrium bid vector, $\tilde{\mathbf{b}}$, each agent's allocation is bounded above as follows:
	\begin{align*}
		x_i(\tilde{\mathbf{b}}) < 1-1/\alpha.
	\end{align*}
\end{corollary}

\begin{proof}
    From the characterization of the equilibrium bids (\ref{eq:equilibrium_bids}), $b_i >0$  implies $a(1-x_i)-1 >0 \implies x_i < 1-1/\alpha$. \hfill \qed
\end{proof}

\Cref{fig:two-player-equil} shows the equilibrium allocations and social costs for the two-player game. Again, we fix $c_1=1$ and increase $c_2$. We show various values of $\alpha$ to demonstrate (i) the increase in the curvature of the $\alpha$-PAR allocation rule resulting in faster concentration of the allocation on the lower-cost player, and (ii) the convergence of the allocations to $x_1 \to 1-1/\alpha$ and $x_2 \to 1/\alpha$ in the left subplot (shown as horizontal lines with matching line styles). 
The right subplot shows the scaling of the social cost in $c_2$.

\begin{figure}
    \centering
    \includegraphics[width=\linewidth]{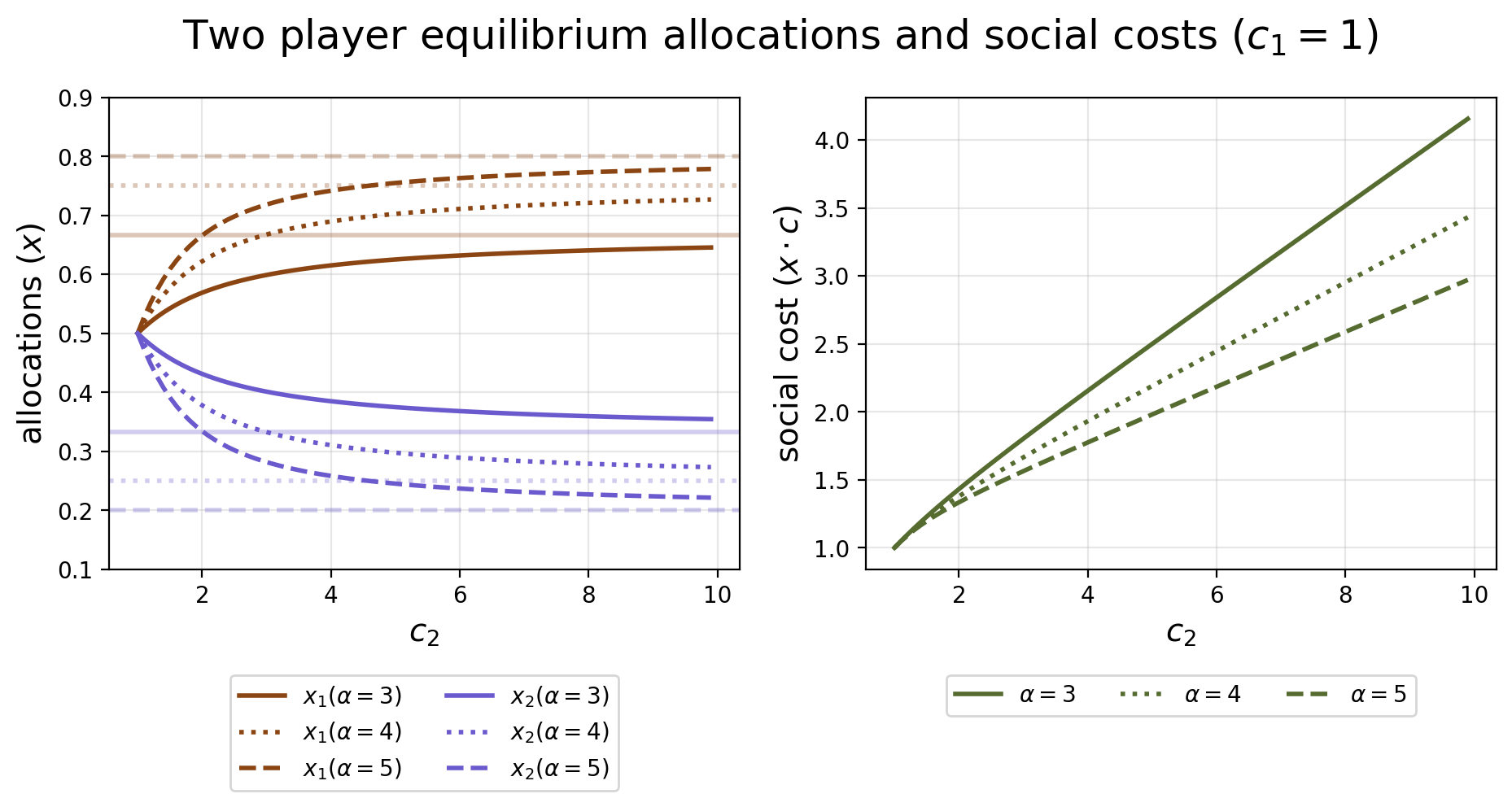}
    \caption{Equilibrium allocations for the two-player game and corresponding social costs. The left plot shows the $1/\alpha$ and $1-1/\alpha$ bounds for the lower and higher allocations, respectively, as horizontal lines.
    }
    \label{fig:two-player-equil}
\end{figure}

Note that, without additional constraints, the paid-as-bid mechanism can result in arbitrarily high payments. 
This is easy to see with \Cref{corr:bounded}; consider the two-player game with $\alpha=3$, $c_1=1$, and $c_2 \to \infty$. We know that $x_1 \to  2/3, x_2 \to 1/3$, and by individual rationality, we must pay player two more than $c_2 \cdot x_2$. Thus, the total payment could be unbounded. As in the DSIC mechanism, this can be resolved by setting a reserve price with a maximal reportable cost of $\bar{b}$. However, this would change the equilibrium characterization given above. Characterizing the new equilibrium is not straightforward, as bidders with costs below $\bar{b}$ might still have bid above $\bar{b}$ at equilibrium. Thus, we leave characterizing the equilibrium with a reserve for future work and focus our analysis on the mechanism without a reserve. 


We now turn our attention to estimating the worst-case scaling of the social cost, which is equivalent to the price of anarchy when we fix $c_1=1$. We characterize the worst-case cost vector $\mathbf{c}$ for the PoA parameterized by the ratio $C = c_{\max}/c_{\min}$ for $\alpha \le 5$ 

\begin{lemma}\label{lemm:worst-case-pab}
    With $n$ bidders where $c_{\max}/c_{\min} \le C$ and with $\alpha \le 5$, the worst-case PoA is realized with the cost vector $[1,C,C,\ldots, C]$.
\end{lemma}

\begin{figure}
    \centering
    \includegraphics[width=0.5\linewidth]{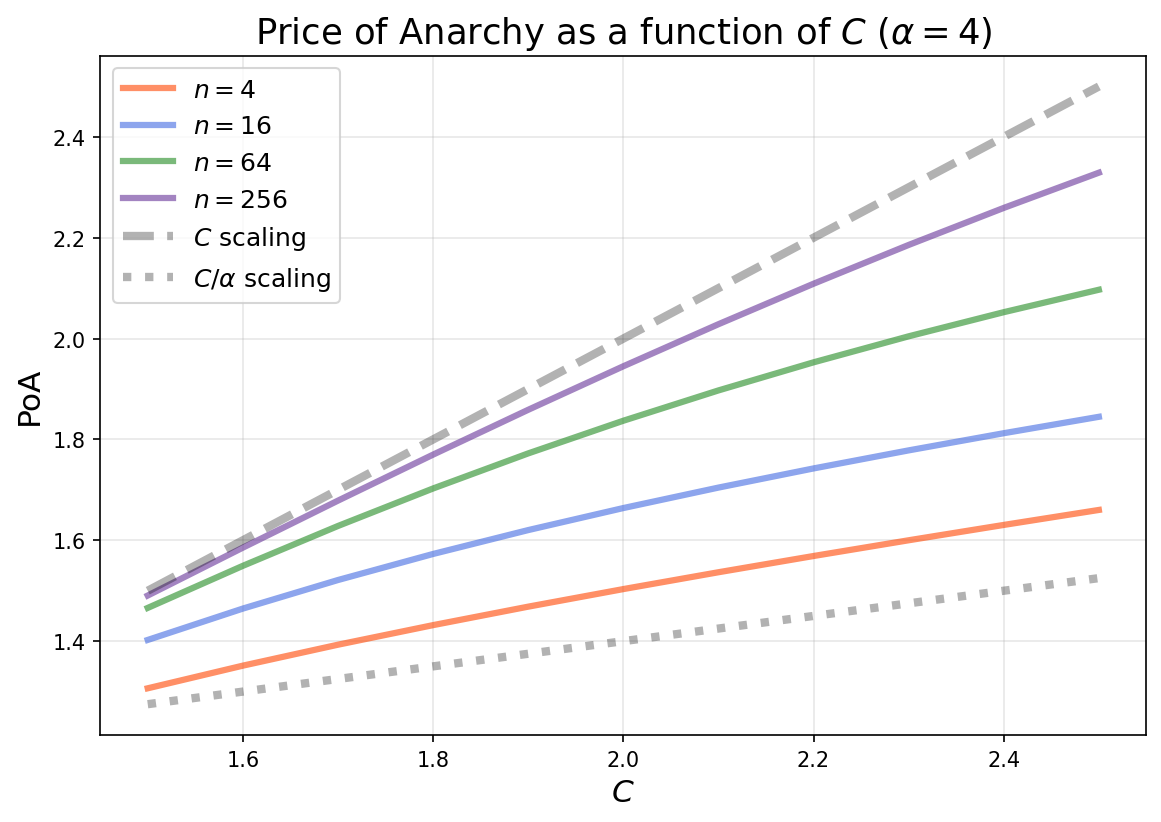}
    \caption{Price of anarchy scaling in $C = c_{\max} /c_{\min}$ for $\alpha=4$ and various values of $n$.}
    \label{fig:poa-c-ca-scaling}
\end{figure}

From \Cref{corr:bounded}, we know the PoA will always be in the range $[C/\alpha, C]$ because the higher-cost players will always be allocated at least $1/\alpha$ in aggregate. We also know that as $n$ increases, the PoA goes from $C/\alpha$ to $C$ because each additional player is allocated a non-zero amount, which pulls some of the allocation from the low-cost player. We use this corollary along with the PNE characterization given in \Cref{eq:equil_condition} to numerically calculate the PoA as a function of $C,\alpha$, and $n$. \Cref{fig:poa-c-ca-scaling} plots this PoA as a function of $C$ for $\alpha=4$ and various values of $n$.
As $n$ increases, the PoA shifts from the $C/\alpha$ scaling regime (shown as the dotted line) to $C$ scaling (shown as the dashed line).  For $\alpha >5$, via the same proof, we have that the worst case cost vector is still of the form $[1,r,\dots,r]$ but for sufficiently large $n$ and a relatively small $C$, the worst case $r$ need not be $C$.

\section{Conclusion}\label{sec:conclusion}

This work serves as a starting point for formalizing procurement mechanisms that do not optimize purely for efficiency. While we explore a range of mechanisms in this article, none is a panacea for all the desiderata presented in \Cref{sec:model}. The DSIC mechanism in \Cref{sec:dsic} has serious drawbacks of not being Sybil-proof and having complex payment rules, each of which is addressed in isolation by the Tullock and the paid-as-bid mechanism of \Cref{sec:tullock,sec:PAB}, respectively. Tullock contests are Sybil-proof at the cost of not being ex-post safe; the paid-as-bid mechanism has a simple and interpretable payment rule structure, but loses Sybil-proofness. The DSIC and paid-as-bid mechanisms, by default, have arbitrarily high payments, but can both be modified to bound the protocol's cost. As such, the immediate domain for future theoretical work is proving the existence or non-existence of a procurement mechanism that is Sybil-proof, and ex-post safe with non-winner-take-all equilibria.

We believe that this work is well-motivated as blockchains continue to mature and develop alternative roles for nodes participating in the network. For example, Ethereum is considering migrating to a ZK-proof-based execution environment, where each block is built along with a proof of validity, and attesters are only required to verify the proof rather than execute the transactions in their entirety. This paradigm shift could introduce the role of ``provers'' explicitly into consensus, and this work shows how different modeling choices could lead to different mechanisms. If the cost of generating a proof continues to drop (it is currently estimated to be about 5 cents per block) and is publicly known (e.g., through proof data aggregation layers like https://ethproofs.xyz/), then either the Tullock contest or the paid-as-bid mechanisms might be the best choice. When considering the possibility of a Sybil gadget (e.g., staking) that overlays the actual procurement mechanism, and the explicit goal of making proving decentralized, choosing an $\alpha$-PAR might be well motivated. We hope that practitioners and builders contribute to the dialogue and help narrow down the right modeling decision for decentralized procurement.
\paragraph*{Model extensions}
While $\alpha$-PARs allow the designer to trade off between distribution of the allocation and efficiency, it remains a subject of debate as to the exact metric of fairness to optimize for. \Cref{lemma:opt} demonstrates that $\alpha$-PARs minimize a scaled social cost, but there are other notions of fairness (e.g., Nash Social Welfare \cite{nash1950bargaining}) that may be of interest. More generally, micro-founding fairness in blockchain settings remains important in justifying the value of decentralization that is often cited as the goal of existing systems. Whether the decentralization offers unique properties to the system (e.g., censorship resistance) or aims to achieve better user outcomes (e.g., more competition and thus lower prices), it is often accepted as an unquestioned axiom in blockchain protocol design without proper justification.

Extending the model to the Bayesian setting is also an open problem. We justify the use of the complete-information setting and the study of PNE because the open, repeated nature of the auctions may allow providers to learn about their competitors' private costs. While this may be true, modeling this knowledge in the incomplete-information setting as defining a prior over the competitors' costs and studying the Bayes-Nash Equilibria would yield different outcomes that are worth considering. Further, while we use the repeated nature of the game to motivate trading off efficiency for a more distributed allocation, explicitly modeling the dynamics in a repeated game would be an important way to sanity-check the effectiveness of the mechanisms we propose.

\paragraph*{Empirical extensions}
The blockchain industry has well-established norms for open-source software and public data, which affords excellent visibility into mechanisms being implemented in production today with real money on the line – an excellent empirical testbed. As the nascent prover markets (e.g., Succinct, RISC Zero, and =Nil;) evolve, measuring the practical bidding strategies and resulting allocations will provide a key empirical grounding for this work. Measurements for the number of bidders participating, the size of available rewards, and the hardware costs will all inform the theoretical study of permissionless procurement mechanism design. This data will also help ground which of the presented mechanisms may be most practical. For example, if it turns out that the prover set is relatively small and non-anonymous, using a non-Sybil-proof mechanism (or adding a Sybil-proofness gadget like staking) may be good enough in practice.


\bibliographystyle{splncs04}
\bibliography{aft}

\appendix


\newpage

\section{Supplementary Material}

\subsection{\Cref{sec:dsic} Proofs}\label{app:dsic-proof}

\begin{lemma}[Optimization problem]
	The allocation rule $x_i = c_i^{-\alpha} / \sum_{j=1}^n c_j^{-\alpha}$ is the solution to the following constrained optimization problem,
	\begin{align*}
		\min_{\mathbf{x} \in [0, 1]^n}&\sum_{i=1}^n c_ix_i^{1+1/\alpha}\\
		\text{s.t.,} &\sum_{i=1}^n x_i = 1.
	\end{align*}
\end{lemma}

\begin{proof}
Let
\begin{align*}
F(\mathbf{x}) = \sum_{i=1}^n \frac{c_i x_i^{1+1/\alpha}}{1+1/\alpha}.
\end{align*}
With the constraint $\sum_{j=1}^n x_j = 1$, we have the Lagrangian,
\begin{align*}
    \mathcal{L}(\mathbf{x}) &= F (\mathbf{x}) + \lambda \left(1-\sum_{j=1}^n x_j\right) \\ 
    &= \sum_{i=1}^n \frac{c_i x_i^{1+1/\alpha}}{1+1/\alpha} + \lambda \left(1-\sum_{j=1}^n x_j\right).
\end{align*}
Taking partials, we have
\begin{alignat*}{2}
    \frac{\partial \mathcal{L}}{\partial x_i} &= c_i x_i^{1/\alpha} - \lambda = 0 &&\implies x_i = \left(\frac{\lambda}{c_i}\right)^\alpha 
\end{alignat*}
Combining these with the full allocation constraint, we have
\[
    \sum_{j=1}^n \left(\frac{\lambda}{c_j}\right)^\alpha = 1 \implies \lambda = \left(\frac{1}{\sum_{j=1}^n c_j^{-\alpha}}\right)^{1/\alpha}.
\]
Solving for the final allocation, we have,
\begin{align*}
    x_i &= \left(\frac{\lambda}{c_i}\right)^{\alpha} = \left(\left(\frac{1}{\sum_{j=1}^n c_j^{-\alpha}}\right)^{1/\alpha}/c_i\right)^{\alpha} = \frac{c_i^{-\alpha}}{\sum_{j=1}^n c_j^{-\alpha}}.
\end{align*}
\hfill \qed
\end{proof}

\begin{lemma}[Worst-case social-cost vector]\label{lemm:dsic-worst-case-cost}
    Assuming $\alpha \neq 1$, there exists a $r\ge 1$ such that the worst-case social cost vector is of the form   
    \begin{align*}
        \mathbf{c}_{\text{worst}} = [1, \underbrace{r, r, \ldots, r}_{n-1 \text{ players}}].
    \end{align*}
    \end{lemma}
    \begin{proof}
    Let the cost maximizing vector be given by $\mathbf{c}^*$. Consider the function $g_k(s)$ representing the social cost as bidder $k$'s cost varies, holding the other bidders' costs fixed, i.e. $g_k(s) = SC([c_1,...,c_{k-1},s,c_{k+1},...,c_n])$. Note that for $k>1$, we must have $g_k'(c_i^*) = 0$. Let  $G_{-i}=\sum_{j\neq i} c_j^{-\alpha}$ and $H_{-i} = \sum_{j\neq i} c_j^{1-\alpha}$ and. Then, we can calculate 
    \begin{align*}
        g_i'(x) &= \frac{(1-\alpha)x^{-\alpha}(x^{-\alpha}+G_{-i})+(x^{1-\alpha}+H_{-i})(\alpha x^{-\alpha-1})}{(x^{-\alpha}+G)^2} \\ 
        &= \frac{x^{-2\alpha}+(1-\alpha)G_{-i}x^{-\alpha}+\alpha H_{-i} x^{-\alpha-1}}{(x^{-\alpha}+G_{-i})^2} \\ 
        &=\frac{x^{\alpha-1}(x^{1-\alpha}+(1-\alpha)G_{-i}x+\alpha H_{-i})}{(1+G_{-i}x^\alpha)^2}
    \end{align*}
    Thus for all $i>1$, taking $g'_i(c_i) = 0$ gives us 
    \begin{align*}
        c_i^{1-\alpha}+(1-\alpha)c_i\sum_{j\neq i} c_{j}^{-\alpha}+\alpha\sum_{j\neq i}c_j^{1-\alpha} = 0.
    \end{align*}
    Now let $G=\sum_{j>1} c_j^{-\alpha}$ and $H=\sum_{j>1} c_j^{1-\alpha}$ then we can rewrite 
    \begin{align*}
        g'_i(c_i) = c_i^{1-\alpha}+(1-\alpha)c_i(G-c_i^{-\alpha}+1)+\alpha(H-c_i^{1-\alpha}+1).
    \end{align*}
    Note that $c_i^{1-\alpha} - (1-\alpha)c_i\cdot c_i^{-\alpha}-\alpha c_i^{1-\alpha} =0$ simplifying 
    \begin{align*}
        g'_i(c^*_i) = (1-\alpha)(1+G)c^*_i + \alpha(1+H)=0.    
    \end{align*}
    Since we assume $\alpha \neq 1$ and this holds for all $i>1$, we have that for all $i,j>1$, $c^*_i=c^*_j$. \hfill \qed
\end{proof}

\begin{lemma}
    The worst-case social cost is bounded above by $1+\left(\frac{n}{\alpha}\right)^{\frac{1}{\alpha}}.$
\end{lemma}
\begin{proof}
From \Cref{lemm:dsic-worst-case-cost}, we have the cost vector $[1, r, r, \ldots, r]$ for some $r>1$. Thus, let
\begin{align*}
    f(r)=\frac{1+(n-1)r^{1-\alpha}}{1+(n-1)r^{-\alpha}}.    
\end{align*}
Then, we seek to maximize $f(r)$ over all $r\ge 1$. Taking the derivative of $f$ with respect to $r$,
\begin{align*}
    f'(r) = -\frac{(n-1)(r^\alpha(r(\alpha-1)-\alpha)-r(n-1))}{r(r^a+n-1)^2}.
\end{align*}
Simplifying $f'(r)=0$ then gives us that the maximum is achieved at $r=x^*$, where $x^*$ is the unique solution to 
\begin{align*}
    B(x)=\alpha+x(1-\alpha)+(n-1)x^{1-\alpha}=0,    
\end{align*}
Now  let $T(r) = (n-1)\,r^{\,1-\alpha}$. Then we have $f(r) = \frac{1+T(r)}{1+T(r)/r}$. Furthermore note, 
\begin{align*}
    B(x^*)=0 \implies T(x^*)+\alpha = x^*(\alpha-1) \implies 1+T(x^*)/x^* = \alpha(1+T(x^*))/(\alpha+T(x^*)).
\end{align*}
Plugging this into our expression for $f(r)$ as a function of $T(r)$ in turn gives 
\begin{align*}
    f(x^*) = \frac{\alpha+T(x^*)}{\alpha} = 1 + \frac{T(x^*)}{\alpha}.
\end{align*}
Then note $T(r)$ is decreasing in $r$ since $\alpha >1$. Thus, a lower bound on $x^*$ gives us an upper bound on $T$ and hence $f(x^*)$. We claim $x^*\ge \left(\frac{n}{\alpha-1}\right)^{\frac{1}{\alpha}}$. To see this note that $B\left(\left(\frac{n}{\alpha-1}\right)^{\frac{1}{\alpha}}\right) > 0$ and $B(x)$ is monotone decreasing in $x$.
Hence we have 
\begin{align*}
    f(x^*) \le 1 + \frac{n-1}{\alpha}\left(\frac{n}{\alpha-1}\right)^{\frac{1-\alpha}{\alpha}} \le 1+ \left(\frac{n}{\alpha}\right)^\frac{1}{\alpha}.
\end{align*}
\hfill \qed
\end{proof}

\subsection{\Cref{sec:PAB} Proofs}\label{app:pab-proofs}

We prove that the worst-case PoA under the paid-as-bid mechanism is realized at the cost vector $[1, C, C, \ldots, C]$. 

\begin{lemma}[Worst-case cost vector]\label{prop:worst-case-pab}
With $n$ bidders where $c_{\max}/c_{\min} \le C$ and $\alpha \le 5$, the worst-case PoA under the paid-as-bid mechanism is realized at the cost vector $[1, C, C, \ldots, C]$.
\end{lemma}

\begin{proof}
Let $\varphi := f^{-1}$ denote the inverse of $f(x)=x^{1/\alpha}(1-x)/(\alpha(1-x)-1)$ (well-defined since $f$ is strictly increasing on $[0, 1-1/\alpha)$). At the unique equilibrium (\Cref{thm:unique_pne}), allocations satisfy $c_i f(x_i) = c_j f(x_j)$ for all $i,j$. Let $\lambda$ denote this common value, so that $x_i = \varphi(\lambda/c_i)$ and
\begin{equation}\label{eq:feasibility}
\sum_{i=1}^n \varphi(\lambda/c_i) = 1.
\end{equation}
The equilibrium social cost is then $H(\mathbf{c}) = \sum_{i=1}^n c_i \varphi(\lambda/c_i)$, where $\lambda$ is determined implicitly by (\ref{eq:feasibility}). Since the equilibrium allocations are invariant under uniform scaling of $\mathbf{c}$, we normalize $c_1 = 1$ and constrain $1 \le c_i \le C$ for $i \ge 2$ with $\max_i c_i = C$. We show that any maximizer of $H$ over this feasible set must have $c_i = C$ for all $i \ge 2$. We start by deriving the derivative of $H$ with respect to $c_i$ for each $i$. 

Differentiating (\ref{eq:feasibility}) with respect to $c_k$ (for $k \ge 2$) and solving gives
\begin{equation}\label{eq:dlambda}
\frac{\partial \lambda}{\partial c_k} = \frac{\lambda \, \varphi'(\lambda/c_k)/c_k^2}{\sum_{i=1}^n \varphi'(\lambda/c_i)/c_i}.
\end{equation}
Then, differentiating $H(\mathbf{c}) = \sum_i c_i \varphi(\lambda/c_i)$ with respect to $c_k$:
\begin{equation}\label{eq:dH-raw}
\frac{\partial H}{\partial c_k} = x_k + \left(\sum_{i=1}^n \varphi'(\lambda/c_i)\right) \frac{\partial \lambda}{\partial c_k} - \frac{\lambda}{c_k} \varphi'(\lambda/c_k).
\end{equation}
Using $\varphi'(y) = 1/f'(\varphi(y))$ and $\lambda/c_i = f(x_i)$, define
\[
A := \sum_{i=1}^n \frac{1}{f'(x_i)}, \qquad B := \sum_{i=1}^n \frac{f(x_i)}{f'(x_i)}, \qquad \gamma := \frac{A}{B}.
\]
Substituting (\ref{eq:dlambda}) into (\ref{eq:dH-raw}) and using $\gamma = A/B$:
\[
\frac{\partial H}{\partial c_k} = x_k + \gamma \cdot \frac{f(x_k)^2}{f'(x_k)} - \frac{f(x_k)}{f'(x_k)}.
\]
Multiplying both sides by the positive factor $\frac{f'(x_k)}{f(x_k)^2}$:

\begin{align*}
    \frac{\partial H}{\partial c_k} \cdot \frac{f'(x_k)}{f(x_k)^2} &= \frac{x_k f'(x_k)}{f(x_k)^2} + \gamma - \frac{1}{f(x_k)} \\&= \gamma - \left(\frac{1}{f(x_k)} - \frac{x_k f'(x_k)}{f(x_k)^2}\right).
\end{align*}

Now define the function $\Psi(x) := \frac{f(x) - xf'(x)}{f(x)^2}$. Then since $f'(x_k) > 0$ in the equilibrium domain, we conclude
\begin{equation}\label{eq:sign-identity}
\operatorname{sign}\left(\frac{\partial H}{\partial c_k}\right) = \operatorname{sign}\bigl(\gamma - \Psi(x_k)\bigr).
\end{equation}

We claim that given $\alpha>1$, $\Psi$ is strictly decreasing for $x\in(0,1-1/\alpha)$. 

\begin{proof}
A direct computation shows that
\[
\Psi'(x) = \frac{x^{-1/\alpha - 1}}{\alpha^2(1-x)^3} \cdot M(x),
\]
where the prefactor $\frac{x^{-1/\alpha-1}}{\alpha^2(1-x)^3} > 0$ for $x \in (0,1)$ and
\[
M(x) = -(\alpha-1)^2(1-x)^2 - \alpha x\bigl((\alpha-1)(1-x^2) + 2\alpha x\bigr).
\]
For $\alpha > 1$ and $x \in (0,1)$, both $(\alpha-1)^2(1-x)^2 > 0$ and $\alpha x\bigl((\alpha-1)(1-x^2) + 2\alpha x\bigr) > 0$, so $M(x) < 0$ and hence $\Psi'(x) < 0$. \hfill \qed
\end{proof}

Now, note the feasible set is compact and $H$ is continuous in $c$, so a maximizer $\mathbf{c}^\star$ exists. Let $x^*$ be the unique equilibrium allocation induced by $c^*$. By the construction of the feasible set we have $c^*_1=1,1\le c^*_i\le C$ for all $i\ge 2$. Let $m:= \max_{i\ge 2} c_i^*$ and $k\in \text{argmax}_{i\ge 2} c_i^*$. We claim that $c_2^*=\dots=c_n^*=m$. 

If $m-1$ then $c_i^*=1$ for all $i\ge 2$ and we're done, so assume otherwise. If there exists a $j\ge 2$ with $c_j^* <m$, then because $\varphi$ is strictly increasing, we have 

$$\frac{\lambda}{c_j^*} > \frac{\lambda}{m} \implies x^*_j=\varphi(\lambda/c^*_j)>\varphi(\lambda/C)=x^*_k$$

Because $\Psi$ is strictly decreasing on $(0,1-1/\alpha)$, this gives $\Psi(x^*_j)<\Psi(x^*_k)$. 

Now, applying first order optimality constraints, if $m\le C$, then $\frac{\partial H}{\partial c_k}(c^*)\ge 0$ implying $\gamma \ge \Psi(x_k^*)$. Then $\Psi(x_j^*)<\Psi(x_k^*)$ implies $\frac{\partial H}{\partial c_k}(c^*)>0$. But $c_j^* < m < C$ so increasing $c_j^*$ is feasible and strictly increases $H$, contradicting $c^*$ being a maximizer. Hence all $i\ge 2$ must satisfy $c_i^*=m$. Implying $c^*$ is of the form $$c(r)=[1,r,\dots,r], \ \ r\in [1,C].$$

We now show that $H(c(r))$ is strictly increasing for $r\in[1,C]$. Let $x(r)$ correspond to the equilibrium allocation under $c(r)$ with $x_1(r)=a(r)$ and $x_2(r)=\dots=x_n(r)=b(r)$. Since $a+(n-1)b=1$, 
$$H(c(r)) = a(r)+r(n-1)b(r)=r-(r-1)a(r),$$ 
and differentiating gives 
$$H'(r)=1-a(r)-(r-1)a'(r).$$

For the ease of exposition, we omit the reference to $r$ where obvious. From the equilibrium condition we have $f(a)=rf(b)$.  Differentiating with respect to $r$ gives 
$$f'(a)a' = f(b)+rf'(b)b' \text{  with  } b'=-\frac{a'}{n-1}.$$ 
Hence 
$$\left(f'(a)+\frac{r}{n-1}f'(b)\right)a'=f(b),$$
so
$$a'(r)=\frac{f(b(r))}{f'(a(r))+\frac{r}{n-1}f'(b(r))}.$$

Then setting $H'(r)>0$ implies 
$$(1-a)\left(f'(a)+\frac{r}{n-1}f'(b)\right)>(r-1)f(b).$$
Using $f(a)=rf(b)$ this is equivalent to 
\begin{equation}\label{eq:h_ineq}
    (1-a)\frac{f'(a)}{f(a)}+b\frac{f'(b)}{f(b)}>1-\frac{1}{r}
\end{equation}

Hence the lemma follows from showing \cref{eq:h_ineq} holds for $\alpha \le 5$. From the definition of $f(x)$, we get 
\begin{equation*}
    \frac{f'(x)}{f(x)}=\frac{1}{ax}+\frac{1}{(1-x)\left(a(1-x)-1\right)}.
\end{equation*}

Inserting this to \cref{eq:h_ineq} gives 

\begin{align*}
    &(1-a)\frac{f'(a)}{f(a)}=\frac{1-a}{\alpha a}+\frac{1}{\alpha(1-a)-1}, \\  &b\frac{f'(b)}{f(b)}=\frac{1}{\alpha}+\frac{b}{(1-b)(\alpha(1-b)-1)}\ge\frac{1}{\alpha}.
\end{align*}

Therefore, 

$$(1-a)\frac{f'(a)}{f(a)}+b\frac{f'(b)}{f(b)}\ge \underbrace{\left(\frac{1}{\alpha a}+\frac{1}{\alpha(1-a)-1}\right)}_{=:\Phi_\alpha(a)}.$$

Since $a < 1-\frac{1}{\alpha}$, $\Phi_\alpha(a)$ is well-defined and positive. Differentiating gives, 

$$\Phi_{\alpha}(a)=-\frac{1}{\alpha a^2}+\frac{\alpha}{\left((\alpha-1)-\alpha a\right)^2}.$$

Solving for $\Phi_\alpha'(a)=0$ gives $a^*=\frac{\alpha-1}{2\alpha}$, so $$\min_a \Phi_\alpha(a) = \frac{4}{\alpha-1}.$$

Therefore, if $\alpha \le 5$, we have $\frac{4}{\alpha-1}\ge 1$ implying 

$$(1-a)\frac{f'(a)}{f(a)}+b\frac{f'(b)}{f(b)}> \Phi_\alpha(a^*)>1-\frac{1}{r}$$ 

in turn showing $H'(r) > 0$ for all $r>1$ from which the lemma follows. 

\end{proof}

\end{document}